\documentclass[pra,superscriptaddress,twocolumn]{revtex4}
\usepackage{epsfig,amsbsy,amsmath,amsfonts,graphics,latexsym}
\usepackage{eepic,bm,psfig,color,subfigure}
\def\bra#1{\mathinner{\langle{#1}|}}
\def\ket#1{\mathinner{|{#1}\rangle}}

\newcommand{\beq}{\begin{equation}}
\newcommand{\eeq}{\end{equation}}

\newcommand{\tr}{{\tt Tr}} 
\newtheorem{propn}{Proposition}{}{}
{}{}
{}{}
{}{}
{}{}
\newenvironment{proof}{{\bf Proof:}}{~\mbox{$\Box$} \\}
\begin{document}
\title{Decoherence-free quantum information in the presence of dynamical
  evolution }   
\author{Peter G. Brooke}
\email{pgb@ics.mq.edu.au}
\affiliation{Centre for Quantum Computer Technology and Department of Physics,
Macquarie University, Sydney, New South Wales 2109, Australia}
\author{Manas K. Patra}
\affiliation{Department of Computing and Mathematics, University of Western
  Sydney, Locked Bag 1797, Penrith South DC, New South Wales 1797, Australia}
  \affiliation{Department of Computer Science, University of York, Heslington, York, United Kingdom YO10 5DD}
\author{James D. Cresser}
\affiliation{Centre for Quantum Computer Technology and Department of Physics,
Macquarie University, Sydney, New South Wales 2109, Australia}
\date{\today}
\begin{abstract}
We analyze decoherence-free (DF) quantum information in the presence of an
arbitrary non-nearest-neighbor bath-induced system Hamiltonian using a
Markovian master equation.  We show that the most
appropriate encoding for $N$ qubits is probably contained within the $\sim
\tfrac{2}{9} N$ excitation subspace.  We give a timescale 
over which one would expect to apply other methods to correct for the system
Hamiltonian.  In order to remain applicable to experiment, we then focus on
small systems, and present examples of DF quantum information for three and
four qubits.  We give an encoding for four qubits that,  while quantum
information remains in the two-excitation subspace, protects against an
arbitrary bath-induced system Hamiltonian.  Although our results are general
to any system of qubits that satisfies our assumptions, throughout the paper we
use dipole-coupled qubits as an example physical system.
\end{abstract}
\maketitle
%%%%%%%%%%%%%%%%%%%%%%%%%%%%%%%%%%%%%%%%%%%%%%%%%%%%%%%%%%%%%%%%%%%%%%%%%%%%%
\section{Introduction}
%%%%%%%%%%%%%%%%%%%%%%%%%%%%%%%%%%%%%%%%%%%%%%%%%%%%%%%%%%%%%%%%%%%%%%%%%%%%
Storing quantum information for long periods is
difficult: excited quantum states decohere.  One method of counteracting the
effects of decoherence is to encode logical information into decoherence-free
subspaces and subsystems
(DFSs)~\cite{Pal96,Duan97,Duan98,Zan97a,Zan97b,Zan98,Lidar98,Knill00}.  
These are groups of states that have robust symmetry
properties, and that in ideal cases serve as perfect quantum memory. See
Ref.~\cite{Lidar03} for a comprehensive review of DFS theory.  As well as
theoretical developments, there has been much success experimentally.  For 
example, DFSs have been prepared in optical systems, in NMR, and in
ion-traps~\cite{Kw00,Kiel01,Lo01,Bou04}, and there has been a proposal
for decoherence-free (DF) quantum-information processing in
nitrogen-vacancy (NV) centers in diamond~\cite{Brooke07a}.          

Here, we focus on the semigroup formulation of DFSs for a number of reasons.
First, the unitary and nonunitary evolution of the system are naturally
separated.    Second, when the DF condition is 
satisfied, there is only one nonunitary Lindblad
operator~\cite{Zan98,Lidar98}.  Finally, the scalability  properties of 
various encodings can be quantified~\cite{Lidar98,Lidar03}.  The separation of
unitary and nonunitary evolution enables the detrimental effect of any
environment-induced unitary evolution to be isolated.  We refer to
`environment-induced', or equivalently `bath-induced', evolution because both
the unitary and nonunitary parts of the master equation result 
directly from the system-environment interaction Hamiltonian.  For a single
atom, the unitary part gives rise to the Lamb-shift and the nonunitary part
gives rise to spontaneous decay.  For many qubits, decoherence can result
indirectly from the action of Lamb-shift type terms because these can cause
quantum information to leak into states that decay. 
A full analysis of the effect of unitary evolution on DF quantum information
in Markovian systems is the purpose of this paper.  
Throughout the paper, when we refer to a system Hamiltonian causing
leakage of DF quantum information, we are referring
explicitly to the action of a Lamb-shift type Hamiltonian on DF states.    

It is well-known that the condition for DF dynamics derived in
Ref.~\cite{Lidar98} does not rule out the possibility of an
environment-induced unitary operator 
evolving information encoded in a DFS into non-DFS states.  The results in
Ref.~\cite{Lidar98} were extended in Ref.~\cite{Shab05}, where the authors
gave conditions for DF quantum information that accounted for a system
Hamiltonian.  Here, we focus on quantum information 
that is encoded within states that satisfy the DF condition~\cite{Lidar98},
but that could evolve into non-DF states due to a system Hamiltonian.  We
emphasize that for 
Markovian systems,  the property of the bath operators is such that generally,
the presence of a non-nearest-neighbor system Hamiltonian is unavoidable.   
Physical justification for the regime studied in this paper is given in
Ref.~\cite{Brooke07b}.  Here, the authors showed that for closely-spaced
dipole-coupled qubits, the nonunitary evolution can be approximated as being
the same as that for co-located qubits (and so satisfy the DF requirements
derived in Ref.~\cite{Lidar98}), but the unitary evolution depends on both
the position and orientation of the dipoles.  Our results are applicable to
any physical system that can be appropriately described by a Markovian master
equation and, in light of recent experimental progress, are directly
applicable to NV centers in diamond~\cite{Da94}. 

The paper is structured as follows.  In Sec.~\ref{sec:setup}, we
state the conditions on the bath operators that lead to a bath-induced
unitary evolution.  In Sec.~\ref{sec:exis}, we give a condition 
for subspaces which are immune to nonunitary evolution.  This condition is
weaker than that derived in Ref.~\cite{Shab05}, but stronger than that derived
in Ref.~\cite{Lidar98}.  Then, in Sec.~\ref{sec:nqs} we analyze the
scalability---which we define as the encoding efficiency multiplied
by the proportion of the Hilbert space that is DF---for DF
encoding in the presence of a system Hamiltonian.  An interesting consequence
of our results is that, when the 
effect of a system Hamiltonian is included, the most suitable subspace for
quantum information storage in $N$ qubits is probably the subspace with $\sim
\tfrac{2}{9} N$ excitations.  Then, we give a timescale over which other 
methods will have to be applied to counter the effect of the system
Hamiltonian. In Sec.~\ref{sec:sqs}, we concentrate on three and four qubit
systems.  We give conditions on the system Hamiltonian for DF
quantum information, and we give an explicit encoding for four qubits that,
while quantum information remains in the two-excitation subspace, protects
against errors induced by an arbitrary system Hamiltonian.  We hope that this
encoding is directly relevant to experimental quantum information processing. 
%%%%%%%%%%%%%%%%%%%%%%%%%%%%%%%%%%%%%%%%%%%%%%%%%%%%%%%%%%%%%%%%%%%%%%%%%%%%%
\section{Master Equation}
\label{sec:setup}
%%%%%%%%%%%%%%%%%%%%%%%%%%%%%%%%%%%%%%%%%%%%%%%%%%%%%%%%%%%%%%%%%%%%%%%%%%%%
For clarity and to establish our notation, we briefly summarize a derivation
of the Lindblad master equation.  A system $A$ coupled to a bath
$B$ can be described by the Hamiltonian $
\mathbf{H} = \mathbf{H}_A \otimes \mathbf{I}_B + \mathbf{I}_A \otimes
\mathbf{H}_B + \mathbf{H}_I$,  where $\mathbf{H}_A$($\mathbf{H}_B$) the system
(bath) Hamiltonian acts on the system (bath) Hilbert space,
$\mathbf{I}_A$($\mathbf{I}_B$) is the identity operator on the system (bath)
Hilbert space, and $\mathbf{H}_I$ is the interaction Hamiltonian that contains
all non-trivial couplings between the system and the bath and is written as
\begin{align}
\label{eq:hint}
 \mathbf{H}_I = \sum_{\alpha} \mathbf{S}_\alpha \otimes \mathbf{B}_\alpha,
\end{align}  
for $\mathbf{S}_\alpha$($\mathbf{B}_\alpha$) the system (bath) operators.
Treating the interaction Hamiltonian as a perturbation, the equation of motion
for the density matrix $\chi$ of the system and bath in the interaction picture
is $(\hbar = 1)$
\begin{align}
\dot{\chi}(t) = -i [ \mathbf{H}_I(t), \chi(t)],
\end{align}
which gives
\begin{align}
\label{eq:exact}
\dot{\rho}(t) =  - \int_0^t \text{d}s
\tr_B [\mathbf{H}_I(t), [ \mathbf{H}_I(s), \chi(s)]], 
\end{align}
where $\tr_B$ denotes the trace over the bath and $\rho$ is the (reduced)
system density matrix.  So, for the rest of this paper, we focus on the
Hilbert space of the system.  Making the Born-Markov approximations,
Eq.~\eqref{eq:exact} becomes 
\begin{align}
\dot{\rho}(t) = - \int_0^\infty \text{d}s \tr_B [\mathbf{H}_I(t), [
    \mathbf{H}_I(s), \rho(t) \otimes R_0]] 
\end{align}
for $R_0$ the initial density operator of the (stationary) bath.
Introducing the correlation function
\begin{align}
C_{\alpha \beta}(s) \equiv
\tr_B [\mathbf{B}^\dagger_{\alpha}(s)\mathbf{B}_{\beta}(0) R_0],
\end{align}
which we note has real and imaginary
parts because in general $\mathbf{B}^\dagger_{\alpha}(s)$ and 
$\mathbf{B}_{\beta}(0)$ do not commute, the master equation in the
rotating-wave approximation is written~\cite{Lind}  
\begin{align}
\label{eq:lme}
\dot{\rho} &= \text{L}[\rho] =
-i[\mathbf{H}_{S}, \rho ]
+ \text{L}_{\text{D}}[\rho], \nonumber \\
\text{L}_{\text{D}}[\rho] &= 
\sum_{\alpha,\beta} a_{\alpha\beta}
\text{L}_{\mathbf{S}_{\alpha},\mathbf{S}_{\beta}} [\rho], \nonumber \\ 
\text{L}_{\mathbf{S}_{\alpha},\mathbf{S}_{\beta}} [\rho] &= 
   [\mathbf{S}_{\beta}, \rho
   \mathbf{S}_{\alpha}^{\dagger}] 
   + [\mathbf{S}_{\beta}\rho,
   \mathbf{S}_{\alpha}^{\dagger}]
 , 
\end{align}
for the system Hamiltonian 
\begin{align}
\mathbf{H}_S =  \sum_{\alpha,\beta}
b_{\alpha \beta} \mathbf{S}_{\alpha}^{\dagger}
\mathbf{S}_{\beta},
\end{align}  
and where $a_{\alpha \beta} = \Gamma_{\alpha   \beta} +   \Gamma^*_{
  \beta   \alpha}$ and $b_{\alpha \beta} = \frac{1}{2i}[ \Gamma_{\alpha
  \beta} - \Gamma^*_{ \beta \alpha}]$ for 
\begin{align} 
\Gamma_{\alpha \beta} \equiv \int_0^\infty \text{d}s
\text{e}^{i \omega_0 s} C_{\alpha \beta}(s).
\end{align}
The correlation function $C_{\alpha \beta}(s)$ contains all the information
about the physical system, and satisfies the Kramers-Kronig relations. 
For a single-atom, $\mathbf{H}_S$ describes the Lamb-shift, and
$\text{L}_{\text{D}}[\rho]$ describes the spontaneous emission.
For two or more dipole-coupled qubits, the off-diagonal terms in
the Hermitian matrix $(a_{\alpha \beta})$ describe the rate of spontaneous
emission between separate physical qubits, and the off-diagonal terms in the 
matrix $(b_{\alpha \beta})$ describe the coherent dipole-dipole interaction
between separate qubits.  Generally, the terms in $(b_{\alpha \beta})$ diverge
and require renormalization~\cite{Brooke07b}.
%%%%%%%%%%%%%%%%%%%%%%%%%%%%%%%%%%%%%%%%%%%%%%%%%%%%%%%%%%%%%%%%%%%%%%%%%%%%%
\subsection{Decoherence-free subspaces}
\label{sec:dfsc}
%%%%%%%%%%%%%%%%%%%%%%%%%%%%%%%%%%%%%%%%%%%%%%%%%%%%%%%%%%%%%%%%%%%%%%%%%%%%
A subspace $\tilde{\rho}$ of the system Hilbert space $\rho$ is a DFS if it
satisfies $\text{L}_{\text{D}}[\tilde{\rho}] = 0$~\cite{Lidar98}.  The
condition that 
leads to $\text{L}_{\text{D}}[\tilde{\rho}] = 0$ being satisfied is
$a_{\alpha\beta} \equiv a$~\cite{Lidar02a}.   For dipole-coupled
qubits, this implies that the spontaneous emission rate between separated
qubits is the same as that for individual qubits---all qubits experience the
same nonunitary couplings.  Here, we focus on the
effect of  
the system Hamiltonian on $\tilde{\rho}$.  Note that when we refer to the
system Hamiltonian, we are not referring to $\mathbf{H}_A$, but $\mathbf{H}_S$
which results from $\mathbf{H}_I$.  We consider the situation in which
the Hamiltonian $\mathbf{H}_S$ evolves information encoded in a DFS into
non-DFS states~\cite{Lidar98,Bacon99}.  We are interested in the physically
applicable perturbative
regime $(a_{\alpha \beta}) = a \mathbf{X} + 
\epsilon \mathbf{A}$ and $(b_{\alpha  \beta}) = b \mathbf{X} + \epsilon
\mathbf{B}$ 
for $\mathbf{X}$ a matrix of size $(a_{\alpha \beta})$
with all entries equal to one, $\mathbf{A}$ and $\mathbf{B}$ matrices also
of size $(a_{\alpha \beta})$, but with arbitrary entries, and expansion
parameter $\epsilon \ll 1$.  

Specifically, we concentrate on the regime $(a_{\alpha \beta}) = a \mathbf{X}$
and 
$(b_{\alpha  \beta}) = b \mathbf{X} + \epsilon 
\mathbf{B}$.  We do this for two reasons.  First, it has been shown that DFSs
are stable to first order even in the presence of a nonunitary symmetry
breaking perturbation~\cite{Bacon99}.  All else being equal, small changes
in the form of $\text{L}_{\text{D}}[\rho]$ do not prevent infinite-lifetime
quantum information storage.  Second, the strict requirement that all qubits
experience the same 
environment is unlikely to ever be met in practice.  If one perturbs the
system-environment couplings, the perturbation applies to both the
nonunitary and unitary parts of the master equation.     

We illustrate the regime of interest to this paper using dipole-coupled qubits.
For these qubits the appropriate expansion parameter is physical separation
$r$.  So, we concentrate on collections of qubits that satisfy  
\begin{align}
\label{eqs:rexp1}
a_{\alpha \beta} &= a + \mathcal{O}(r^2),\\
b_{\alpha \beta} &= b + b^\prime_{\alpha \beta}(r) + \mathcal{O}(r^2),
\label{eqs:rexp2}
\end{align}
where the system of qubits is such that any $\mathcal{O}(r)$ contribution to 
$a_{\alpha \beta}$ can be neglected.  Eqs.~\eqref{eqs:rexp1} and
~\eqref{eqs:rexp2} are satisfied by  
closely-spaced dipole-coupled qubits---explicit forms for $a_{\alpha \beta}$
and $b_{\alpha \beta}$ are given in Ref.~\cite{Brooke07b}. Note that
$b^\prime_{\alpha \beta}(r)$ can be many orders of magnitude larger 
than $a$.  

We emphasize that the results presented here are applicable to any
Markovian system for which $\mathbf{A} = \mathbf{0}$ and $\mathbf{B} \ne
\mathbf{0}$.  This is the most general case for Markovian systems that satisfy
the DF condition, and is more likely to be realised in the laboratory than
$\mathbf{A} = \mathbf{0}$ and $\mathbf{B} = \mathbf{0}$.  Note that if 
$\mathbf{A} = \mathbf{0}$ and $\mathbf{B} = \mathbf{0}$, then 
$\Gamma_{\alpha \beta}$ is strictly qubit independent, and $\mathbf{H}_S$ does
not evolve information encoded in $\tilde{\rho}$ into non-DF
states~\cite{Zan97a,Zan98}.    

For non-Markovian systems, leakage resulting from Lamb-shift type terms in the
master equation can be accounted for using dynamical decoupling methods, or
`bang-bang' pulses~\cite{Viola99,Wu002,Byrd03,Byrd05}.  Here, we focus on {\it
  passive} error correction that does not require fast and strong pulses and
so is more amenable to experimental implementations.  Also, bang-bang pulses
are inherently non-Markovian, and applying similar
techniques in a Markovian environment is difficult.
%%%%%%%%%%%%%%%%%%%%%%%%%%%%%%%%%%%%%%%%%%%%%%%%%%%%%%%%%%%%%%%%%%%%%%%%%%%%%
\section{Completely-decoherence-free subspaces}
\label{sec:exis}
%%%%%%%%%%%%%%%%%%%%%%%%%%%%%%%%%%%%%%%%%%%%%%%%%%%%%%%%%%%%%%%%%%%%%%%%%%%%
The condition $a_{\alpha \beta} \equiv a$ means that the dissipator can be
written   
\begin{align}
\text{L}_{\text{D}}[\rho] =  \lambda (2 \mathbf{J} \rho \mathbf{J}^\dagger
 - \mathbf{J}^\dagger \mathbf{J} \rho  - \rho \mathbf{J}^\dagger \mathbf{J}),
\end{align}
where $\lambda$ is the only nonzero eigenvalue of $a_{\alpha \beta}$ and
$\mathbf{J}  = \sum_{\alpha}^N \mathbf{S}_\alpha$ for $N$ the number of
qubits.  The jump operator can be written
\begin{align}
\mathbf{J} = \sum_{i=1}^N \hat{\sigma}_{i-} =  \sum_{i=1}^N I \otimes \cdots I
\otimes 
\underbrace{\hat{\sigma}_{-} }_{i^{\text{th}} \text{ qubit}}\otimes I \cdots,
\end{align}
where $I$ is the $2\times 2$ identity operator.  So, $\hat{\sigma}_{ia}$, for
$a= +,-,z$, acts on the $i^\text{th}$ qubit, and satisfies
$[\hat{\sigma}_{iz},\hat{\sigma}_{j\pm}]= \pm 2 \delta_{ij} \hat{\sigma}_{i
  \pm}$, and $[\hat{\sigma}_{i+},\hat{\sigma}_{j-}]= \delta_{ij}
\hat{\sigma}_{iz} $. The operators $\mathbf{J}$, $\mathbf{J}^\dagger$, and
$\mathbf{J}_z$ act on the system Hilbert space $\mathcal{H}\equiv
\otimes^N\mathbb{C}^{2}$.  From the representation theory of $su(2)$,
$\mathcal{H}$ can be decomposed into irreducible
components~\cite{Gilmore}---see Fig~\ref{fig:3quenc} for a three qubit
example.  Each  
irreducible representation (irrep) $V$ is generated by a unique 
  lowest weight vector $v$ satisfying $\mathbf{J}\cdot v \equiv
\mathbf{J}_{-} \cdot v^{-} = 0$.  The vector $v^{-}$ is an eigenvector of
\begin{align}
\mathbf{J}_z = \sum_{i=1}^N \hat{\sigma}_{iz}.
\end{align}
The vector $\mathbf{J}^k\cdot v$, for $k$ applications of $\mathbf{J}$ on
$v$ is an eigenvector of $\mathbf{J}_z$ with eigenvalue $-(N-2k)$. 
The decoherence operator $\mathbf{J}$ causes the system to decay to its lowest
weight.  The DF condition $\text{L}_{\text{D}}[V] = 0$ for $V$ a subspace
implies that $\mathbf{K} \cdot V=0$ for all jump operators
$\mathbf{K}$~\cite{Lidar98}.  
Here, we only have one jump operator $\mathbf{J}$, so the subspace
$V$ satisfies $\mathbf{J} \cdot V=0$, and consists of combinations of lowest
weight vectors.  For example, the states $\ket{\sf b}$ and $\ket{\sf c}$ in
Fig~\ref{fig:3quenc} satisfy $\mathbf{J} \ket{\sf b} = 0$ and  $\mathbf{J}
\ket{\sf c} = 0$.  They are states that do not decay to $\ket{\sf a}$ and so
are DF. The operator $\mathbf{J}^\dagger$ acts like
$\mathbf{J}^\dagger \ket{\sf b} = \ket{\sf e}$ and $\mathbf{J}^\dagger
\ket{\sf c} = \ket{\sf f}$, but does not cause transitions
from $\ket{\sf b}$ to $\ket{\sf f}$, or from $\ket{\sf c}$ to $\ket{\sf e}$.

The definition for DF dynamics~\cite{Lidar03}---a subspace
$W$ of states is DF if any $\rho(0) \in W$  evolves into a state $\rho(t)$
such that the evolution 
map $\rho(0) \rightarrow \rho(t)$ is unitary $\forall \;
t$---is relevant here.  We are interested in the evolution of quantum
information under the action of both $\text{L}_{\text{D}}$ and
$\mathbf{H}_S$.  Ensuring $\text{L}_{\text{D}}[\rho] = 0$ for some subspace is not
sufficient to guarantee DF dynamics for all time.  We call a subspace $W$
completely-decoherence-free (CDF) if it satisfies the following conditions 
\begin{enumerate}
\item
\( \text{L}_{\text{D}}[W] = 0\) and 
\item
\(\rho(t)\in W \quad \forall \; t \).
\end{enumerate}
These conditions ensure that the evolution is unitary.  We note that the result
at the replica symmetric point derived in Refs.~\cite{Zan97a,Zan98} ensures
CDF dynamics, and that the conditions stated here still permit the transfer of
encoded information between states within $W$.  For the purposes of this
paper, when we refer to transfer of encoded information, we are referring to 
the effect of the off-diagonal terms of $\mathbf{H}_S$ in the Clebsch-Gordan
basis on quantum information. A basic characterization of
completely-decoherence-free subspaces is given below.       
\begin{propn} \label{propn:BasicDF}
Let $V$ be the subspace of lowest weight vectors. A
necessary and sufficient condition that $V$ contain a CDF subspace $W$ is
$\mathbf{H}_S \cdot W \subset W$. In particular, $\mathbf{H}_S$ can be
diagonalized in $W$.   
\end{propn}
\begin{proof}
Define $\rho^\prime(t) = e^{i \mathbf{H}_S t} \rho(t) e^{-i \mathbf{H}_S t}$
for 
some Hamiltonian $\mathbf{H}_S$. The equation satisfied by $\rho'$ is   
\beq 
\frac{\partial \rho^\prime}{\partial t} =
\text{L}^\prime_{\text{D}}[\rho^\prime], 
\eeq
where $ \text{L}^\prime_{\text{D}}[\mathbf{J}]=\text{L}_{\text{D}}
[e^{i \mathbf{H}_St}\mathbf{J} e^{-i\mathbf{H}_St}]$.  So,
in this picture $\rho^\prime(0)$ is DF iff
$ \text{L}^\prime_{\text{D}} [\rho^\prime(0)]=0$.  The generic
DFS are spanned by vectors $\ket{x}$ such that
$\mathbf{J}^\prime\ket{x}=e^{i\mathbf{H}_St}
\mathbf{J}e^{-i\mathbf{H}_St}\ket{x}=0$. Let $W$ be the subspace consisting
of all such vectors.  This must be satisfied for all $t$, so $\mathbf{H}_S
\cdot W \subset W$.  Conversely, if this is satisfied then  
 \beq \label{eq:cdfs}
\mathbf{J}e^{-i\mathbf{H}_S t}\ket{x}=0.
\eeq
\end{proof}
This condition is weaker than that derived in Ref.~\cite{Shab05}, but stronger
than that derived in Ref.~\cite{Lidar98}.  Throughout this paper, it is
implicitly assumed that if Eq.~\eqref{eq:cdfs} is satisfied, then other
techniques, such as quantum error correction, are applied to account for
transfer of encoded information between states within a CDFS.

We have observed that DFSs are lowest weight vectors in the
  decomposition of ${\mathbb C}^{2^N}$ into irreps of
  $su(2)$. The weight space $W(k)$ of weight $k$ in $\mathcal{H}$ is the
  eigenspace of $\mathbf{J}_z$ with eigenvalue $k$. Since $[\mathbf{J}_z,
  \mathbf{H}_S]=0$ weight spaces $W(k)$  are left invariant by
  $\mathbf{H}_S$.  Hence, it is sufficient to consider subspaces with a fixed
  weight, that is, all states in the subspace have the same number of excited
  (physical) qubits.  Note that these nontrivial
  subspaces span across several irreps, so it is necessary to combine irreps
  if one is to satisfy Eq.~\eqref{eq:cdfs}.  From the representation theory of
  $su(2)$~\cite{Gilmore}, it follows that the weight spaces have weights
  $-N,-(N-2), \dotsc, 
  N$.  For $k\leq N/2$ the dimension of space $W(k)$ for $k$ excitations with
  weight $-(N-2k)$ is $\binom{N}{k}$.  The condition $\mathbf{J} \cdot V
  (k)=0,\; 
  V(k)\subset W(k)$ allows us to write the dimensions of $V(k)$ as
\begin{align} \label{eq:dimDFS}
\text{dim}[V(k)]= \binom{N}{k}-\binom{N}{k-1}=\frac{N!(N-2k+1)}{k!(N-k+1)!}.
\end{align}
%%%%%%%%%%%%%%%%%%%%%%%%%%%%%%%%%%%%%%%%%%%%%%%%%%%%%%%%%%%%%%%%%%%%%%%%%%%%%
\section{Scalability properties and leakage timescale}
\label{sec:nqs}
%%%%%%%%%%%%%%%%%%%%%%%%%%%%%%%%%%%%%%%%%%%%%%%%%%%%%%%%%%%%%%%%%%%%%%%%%%%%
A sufficient condition for CDF dynamics derived in
Ref.~\cite{Zan97a} for $N$ qubits that are
prepared in a DFS is as follows.  If the unitary coupling satisfies
$b_{\alpha\beta} \equiv b$, then the Hamiltonian can be written 
\begin{align}
\mathbf{H}_S = b \sum_{\alpha,\beta} \mathbf{S}_{\alpha}^{\dagger}
\mathbf{S}_{\beta} = b \mathbf{J}^\dagger \mathbf{J}.
\end{align}
Thus, $\mathbf{H}_S$ is a product of the total operators $\mathbf{J}$ and
$\mathbf{J}^{\dagger}$, and the irreps are left invariant by $\mathbf{H}_S$. 
An example of this is $N$ co-located qubits that are described by 
the master equation derived in Ref.~\cite{Brooke07b}.  In this instance, the
multi-qubit level shift is the same for all qubits. Note that if
$b_{\alpha\beta} \equiv b$ is satisfied, then DFSs are stable to a symmetry
breaking perturbation~\cite{Bacon99}.

For $N$ qubits, one can estimate the effect of $\mathbf{H}_S$ on encoded
quantum information by
examining the proportion of the Hilbert space that consists of DF states 
relative to the proportion that consists of non-DF states. 
We consider vectors of the same weight, or equivalently, the same level of
excitations.  For $V(k)$ the DF subspace in $W(k)$ (with weight $-(N-2k),\;
2k\leq N$) the encoding efficiency is defined as the number of logical qubits
per number of physical qubits.  So, using Eq.~\eqref{eq:dimDFS} we define   
\begin{align}
d_{\text{DF}} \equiv \frac{1}{N}\log_2{\text{dim}[V(k)]} = \frac{1}{N}\log_2
\left[\frac{N!(N-2k+1)}{k!(N-k+1)!} \right].
\end{align}
The encoding efficiency $d_{\text{DF}}$ measures how many DF qubits can be
encoded into a Hilbert space $\mathcal{H}\equiv \otimes^N\mathbb{C}^{2}$, and
is unity for scalable encoding. Writing $k=rN$ where $r\leq 1/2$, and 
taking the limit $N \to \infty$ gives
\begin{align}
\label{eq:ddf}
d_{\text{DF}} \stackrel{\scriptsize{N \to
    \infty}}{\longrightarrow} -r\log_2 r - (1-r)\log_2 (1-r), 
\end{align}
where $N$ is the number of physical qubits, and $r$ is independent of $N$.  For $N \to \infty$,
$d_{\text{DF}}$ is a maximum for $r = 1/2$.  This is the canonical
strong-collective DFS~\cite{Kempe01}.

Encoding efficiency alone is not sufficient as a measure of scalability when
the Hamiltonian $\mathbf{H}_S$ can cause leakage of quantum information from
DF to non-DF states.  So, as a measure of the
likelihood that an arbitrary $\mathbf{H}_S$ causes quantum information to
transfer between states within a particular excitation, we define 
\begin{align}
p_{\text{DF}} \equiv \frac{\dim [V(k)]}{\dim[ W(k)]} = 1 + \frac{k}{k - 1 - N},
\end{align}
which is simply the fraction of a particular weight space that satisfies
$\mathbf{J} \cdot V(k)=0$.  So, $p_{\text{DF}}$ measures the proportion of DF
states relative to non-DF states for a particular excitation.   For $N \to
\infty$, $p_{\text{DF}}$ is   
\begin{align}
\label{eq:pdf}
p_{\text{DF}}  \stackrel{\scriptsize{N \to
    \infty}}{\longrightarrow} \frac{1 - 2r}{1 - r}.
\end{align}
\begin{figure}[t]
\begin{center} 
\subfigure[]{\label{fig:scal1}
\includegraphics[width=3.8cm,height=3cm]{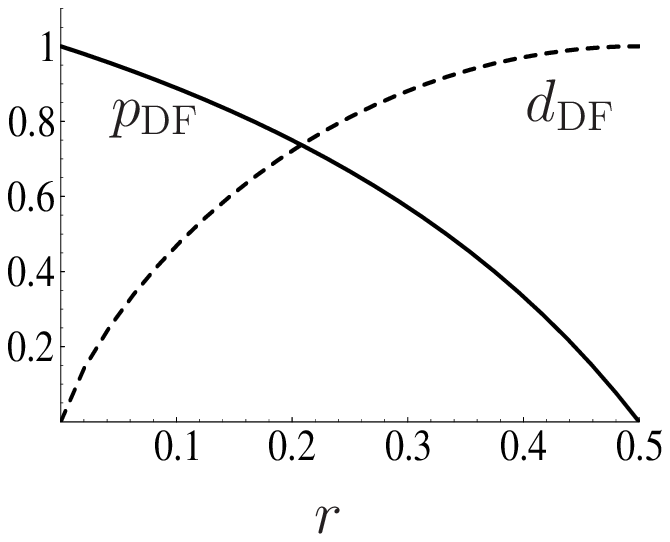}}
\subfigure[]{\label{fig:scal2}
\includegraphics[width=3.8cm,height=3cm]{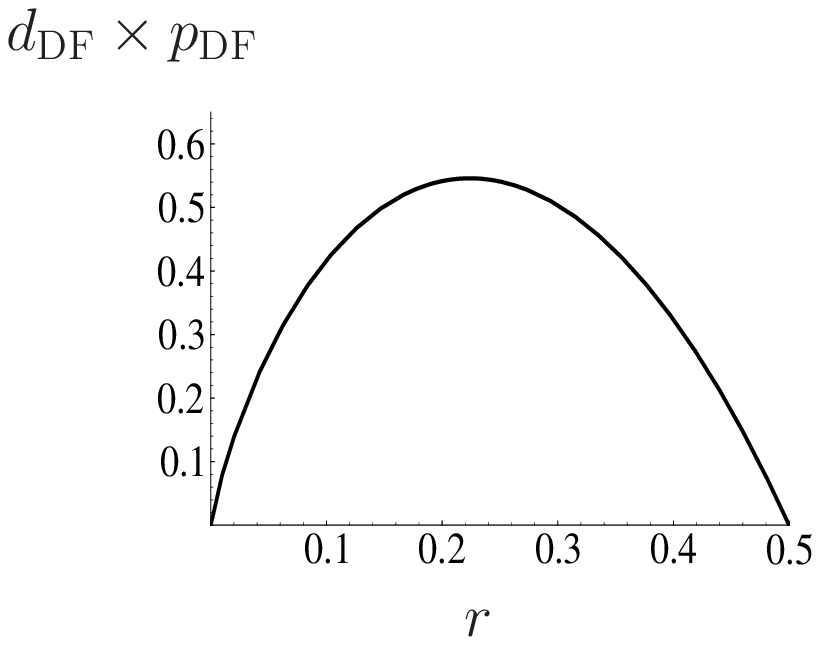}}
\end{center}
\caption{(a) Plot of $d_{\text{DF}}$ (dotted line) and $p_{\text{DF}}$ (solid
  line) and 
  (b) $d_{\text{DF}} \times p_{\text{DF}}$ for large $N$ for $d_{\text{DF}}$
  and $p_{\text{DF}}$ as 
  defined in Eqs.~\eqref{eq:ddf} 
  and~\eqref{eq:pdf} respectively.  In (b), $d_{\text{DF}} \times
  p_{\text{DF}}$ is a maximum for $r \sim \tfrac{2}{9}$.
\label{fig:nqscal} } 
\end{figure}
The smaller $p_{\text{DF}}$, the more likely that $\mathbf{H}_S$ will 
cause quantum information to evolve into non-DF states in that particular
weight space. We show $d_{DF}$ and $p_{DF}$ for large $N$ in
Fig.~\ref{fig:scal1}, and the product $d_{DF} \times p_{DF}$ in
Fig.~\ref{fig:scal2}.  For an 
arbitrary $\mathbf{H}_S$ that causes quantum
information to transfer between states within some weight space, the subspace
which maximises $d_{DF} \times p_{DF}$ is one for which $k \sim 
\tfrac{2}{9} N$.   
Of course, this assumes that one can account for the unitary evolution
caused by $\mathbf{H}_S$ within the CDFSs.  For particular forms of
$\mathbf{H}_S$, it might 
be more appropriate to encode in other weight subspaces, but if
scalability in this context is important, then care must be taken to ensure
$d_{DF} \times p_{DF}$ is large.  

As an example, we focus on the case of strong-collective
decoherence~\cite{Kempe01}.  The DF subspace in this instance has
dimension~\cite{Zan97a,Lidar98,Lidar03}   
\begin{align}
\dim[\text{DFS}(N)] = \frac{N!}{(N/2 +1)!(N/2)!},
\end{align}
for the collective basis $\ket{J,m_J}$, where $\ket{0}$($\ket{1}$) represents
  a $\ket{j =   \tfrac{1}{2},m_j =   -\tfrac{1}{2}}$($\ket{j =
  \tfrac{1}{2},m_j =  \tfrac{1}{2}}$) state. Note that we are not referring
  here to the physical angular momentum of a particle, as in
  Ref.~\cite{Kempe01} we are simply using the notation for convenience. For
  large $N$, the encoding efficiency   $d_{\text{DF}}$ is asymptotically unity.
  However, for $J = 0$ and $r= 
  1/2$, $p_{\text{DF}}  \stackrel{\scriptsize{N \to
  \infty}}{\longrightarrow} 0$.  So, the proportion of DF states in the $N/2$
  subspace is asymptotically zero, 
implying that an arbitrary $\mathbf{H}_S$ that causes information to leak from
  DF-states to non-DF states negates the encoding-efficiency of
  strong-collective DF 
subspaces.  We consider three cases: (i) $\mathbf{H}_S$ causes DF information to
  leak into all other states in
the $N/2$ subspace, (ii) $\mathbf{H}_S$ causes DF information to leak
  from $J=0$ to $J=1$ and $J=2$ states, and
(iii) $\mathbf{H}_S$ causes DF information to leak from $J=0$ to $J=1$
  states.  We define 
\begin{align}
\label{eq:pdfj}
p_{\text{DF},J_{\text{tot}}} \equiv
\frac{\dim[\text{DFS}(N)]}{\dim[\text{non-DFS}(N)]} ,
\end{align}  
for 
\begin{align}
\label{eq:den}
\dim[\text{non-DFS}(N)] = \sum_{J = 1}^{J_{\text{tot}}} \frac{(2J + 1)N!}{(N/2
  + J +1)!(N/2
 -J)!}.
\end{align}
\begin{figure}[t]
\begin{center} 
\includegraphics[width=5.5cm,height=4.cm]{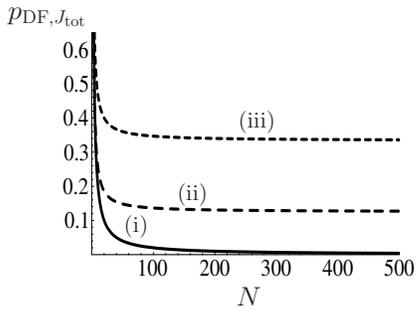}
\end{center}
\caption{Ratio $p_{\text{DF},J_{\text{tot}}}$ for (i) $J = 1, \ldots N/2$,
  (ii) $J = 1,2$, 
  and (iii) $J=1$ for $N$
  qubits, where $p_{\text{DF},J_{\text{tot}}}$ is defined in
  Eq.~\eqref{eq:pdfj}. \label{fig:dfs500} }  
\end{figure}
Note that the quantity
$p_{\text{DF},J_{\text{tot}}}$ equals unity for equal amounts of DF and non-DF
states, whereas for the same instance $p_{\text{DF}}$ equals $1/2$.
Fig.~\ref{fig:dfs500} shows $p_{\text{DF},J_{\text{tot}}}$ for $J_{\text{tot}}
= 1$, $J_{\text{tot}} = 2$, and $J_{\text{tot}} = N/2$ for 500 qubits.
Allowing quantum information to transfer to just one other subspace reduces the
encoding efficiency.  So, understanding the effect of $\mathbf{H}_S$ in
particular physical realizations 
is important for scalability. 

The effect of variations away from $b_{\alpha\beta} \equiv b$ can be
quantified as follows.  The fidelity $F(t) = \text{Tr}[ \rho_U(t) \rho(t)]$
for $\rho_U(t)$ the unwanted unitary evolution and
$\rho(t)$ the desired evolution, is a measure of the effect of the unitary
evolution on encoded quantum information.  If
$\rho_U(t) = \rho(t)$, then $F(t) = 1$ (for pure states), 
and the system serves as a perfect quantum memory.  The fidelity can
be expanded as
\begin{align}
F(t) = \sum_n \frac{1}{n!} \left( \frac{t}{\tau_n} \right)^n,
\end{align}
for
\begin{align}
\left( \frac{1}{\tau_n} \right)^n = \text{Tr}[ \{\rho_U(t) \rho(t) \}^{(n)}], 
\end{align}
where the superscript $(n)$ denotes $n^{\text{th}}$
derivative~\cite{Lidar98,Bacon99}.  The timescale 
$\tau_1^{-1} = 0$ for any $\mathbf{H}_S$, so we focus on $\tau_2^{-1}$.  This
is an estimate of the timescale over which quantum error correction (or some
other technique, eg. a corrective pulse sequence) will have
to be applied~\cite{Neil00}.  So,   
\begin{align}
\label{eq:tsc}
\frac{1}{2}\left( \frac{1}{\tau_2} \right)^2 = \bra{\psi}\mathbf{H}_S
\ket{\psi}^2 - \bra{\psi}\mathbf{H}^2_S \ket{\psi} ,
\end{align}
for $\ket{\psi}$ within a DFS.  We have assumed
the system begins in a pure state in a DFS and that the desired evolution
satisfies $\rho(t) = \rho(0)$.  We are interested in the transfer timescale,
so we assume that $\text{L}_{\text{D}}[\rho_U(t)]
\ne 0$ only for $t > \tau_2$.  If the states
$\ket{\psi}$ are eigenstates of $\mathbf{H}_S$, then $\tau_2^{-1} = 0$.   The
timescale in Eq.~\eqref{eq:tsc} also applies to transferring quantum
information between states within a DFS.         
%%%%%%%%%%%%%%%%%%%%%%%%%%%%%%%%%%%%%%%%%%%%%%%%%%%%%%%%%%%%%%%%%%%%%%%%%%%%%
\section{Small qubit systems}
\label{sec:sqs}
%%%%%%%%%%%%%%%%%%%%%%%%%%%%%%%%%%%%%%%%%%%%%%%%%%%%%%%%%%%%%%%%%%%%%%%%%%%%
In practical applications and in many theoretical proposals, the number of
qubits that are controlled in order to process quantum information is small.
Here, we examine in detail systems that might be realizable in a laboratory,
and give conditions on $b_{\alpha\beta}$ that lead to robust quantum
information storage.  We focus on three and four qubits because of the
possibility that $b_{\alpha \beta}$ is spatially dependent (which is the case
for dipole-coupled qubits).  If so,
experimental control of $b_{\alpha  \beta}$ can be
obtained through varying the spatial arrangement of the qubits. 
%%%%%%%%%%%%%%%%%%%%%%%%%%%%%%%%%%%%%%%%%%%%%%%%%%%%%%%%%%%%%%%%%%%%%%%%%%%%%
\subsection{Three qubits}
\label{sec:thqu}
%%%%%%%%%%%%%%%%%%%%%%%%%%%%%%%%%%%%%%%%%%%%%%%%%%%%%%%%%%%%%%%%%%%%%%%%%%%%
Three qubits is the smallest number of qubits that supports a DF
qubit~\cite{Knill00}.   The eigenbasis is given in Fig.~\ref{fig:3quenc}.
The qubit is encoded as
\begin{align}
\ket{1}_L&= \left \{ \begin{array}{l} 
       \ket{\sf c} = 
       \frac{1}{\sqrt{2}}(\ket{010}-\ket{100}), \label{eq:dicke0L}\\  
       \ket{\sf f} = 
       \frac{1}{\sqrt{2}}(\ket{011}-\ket{101}), \end{array} \right. \\
\ket{0}_L&= \left \{ \begin{array}{l} 
  \ket{\sf b} =
       \frac{1}{\sqrt{6}}(-2\ket{001}+\ket{010}+\ket{100}), 
\label{eq:dicke1L} \\ 
   \ket{\sf e} =
       \frac{1}{\sqrt{6}}(2\ket{110}-\ket{101}-\ket{011}), \end{array}
       \right.  
\end{align}
where the logical groupings are indicated in the first column, the
 basis used in this section is given in the second column, and
the states are expanded in terms of the single-particle basis in the third
column.  Transitions are only allowed between states with the same symmetry,
 i.e., the jump operator $\mathbf{J}$ does not cause quantum information to
 decay from $\ket{\sf e}$ to $\ket{\sf c}$ or $\ket{\sf d}$, or from $\ket{\sf f}$ to $\ket{\sf
 b}$ or $\ket{\sf d}$. It only acts within the logical groupings.  The degeneracy for each $J$
 is given by Eq.~\eqref{eq:den}. 
\begin{figure}[t]
\begin{center} 
\includegraphics[width=6cm,height=4.5cm]{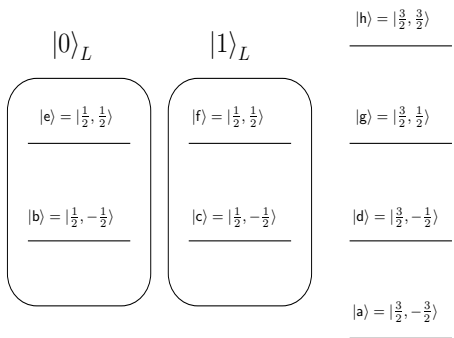}
\end{center}
\caption{DF encoding for three qubits.
  The states are labelled according to $\ket{J,m_J}$. 
  The two isolated subspaces are circled according to the logical basis.  See
  Eqs.~\eqref{eq:dicke0L} and~\eqref{eq:dicke1L} for details.  
\label{fig:3quenc} The splitting of the degeneracy due to the
  non-nearest-neighbor unitary interaction is not  
  included.  } 
\end{figure}

The off-diagonal elements of the Hamiltonian 
\begin{align}
\mathbf{H}_S =  \sum_{\alpha,\beta = 1}^3
b_{\alpha \beta} \hat{\sigma}_{\alpha +}
\hat{\sigma}_{\beta -},
\end{align}
where $\alpha$ and $\beta$ label the qubit, in the one-excitation subspace in
the collective basis are 
\begin{align}
\label{eq:hamod}
\mathbf{H}^{\text{I}}_S =& \frac{1}{\sqrt{6}}(b_{23} - b_{13}) \ket{\sf
  d} \bra{\sf c} + \frac{1}{\sqrt{3}}(b_{13} - b_{23}) \ket{\sf
  c} \bra{\sf b} \nonumber \\
&+ \frac{1}{3 \sqrt{2}}(2b_{12} - b_{13} - b_{23}) \ket{\sf
  d} \bra{\sf b} + \text{H.c.}.
\end{align}
For dipole-coupled qubits, the coefficient $b_{\alpha \beta}$ describes the
dipole-dipole interactions.  One possible physical system that might satisfy
the required conditions are nitrogen-vacancy (NV) centres in diamond.  These
can be manufactured in such a manner that the expansion
detailed in Eqs.~\eqref{eqs:rexp1} and~\eqref{eqs:rexp2} adequately
describes the physical system~\cite{Me05a}. 

A consequence of $[\mathbf{J}_z,\mathbf{H}_S]=0$ is that the Hamiltonian is
block-diagonal in the collective basis.  We examine the three cases: (i)
$b_{\alpha\beta} = b$, (ii) $b_{12} = b_{23} \ne b_{13}$, and (iii) $b_{12} \ne
b_{23} \ne b_{13}$.  It can be seen immediately for $b_{\alpha\beta} = b$ that 
$\mathbf{H}^{\text{I}}_S  = 0$, and Eq.~\eqref{eq:cdfs}
is satisfied.  It can also be seen for $b_{23} = b_{13}$ that $\ket{\sf c}$
is not acted upon by $\mathbf{H}_S$ and is CDF. 

We consider $b_{12} = b_{23} \ne b_{13}$. Introducing the states $\ket{\sf u}
= \tfrac{1}{2} \ket{\sf c} - 
\tfrac{\sqrt{3}}{2} \ket{\sf b}$ and $\ket{\sf v} = \tfrac{\sqrt{3}}{2}
\ket{\sf c} +  
\tfrac{1}{2} \ket{\sf b}$, the off-diagonal system Hamiltonian 
becomes
\begin{align}
\mathbf{H}^{\text{I}}_S = \frac{\sqrt{2}}{3}(b_{23} - b_{13}) \ket{\sf d}
\bra{\sf v} + \text{H.c.},
\end{align}
meaning that $\ket{\sf u}$ is not acted upon by $\mathbf{H}_S$ and is CDF.
So, if 
state $c_{\sf u} \ket{\sf u} + c_{\sf v} \ket{\sf v}$ for $|c_{\sf u}|^2 +
|c_{\sf 
  v}|^2 = 1$ is prepared, the second order transfer rate is 
\begin{align}
\left( \frac{1}{\tau_2} \right)^2 =& 2[b_{13} c_u^2 - \frac{1}{3}(b_{13}- 4
b_{12} )c_v^2]^2 - 2 b_{13}^2 c_u^2 \nonumber \\
&- \frac{2}{3}(6 b_{12}^2 - 4 b_{12} b_{13} +
b_{13}^2 )c_v^2,
\end{align}  
which implies that smaller differences between the elements of
$(b_{\alpha \beta})$ leads to more robust quantum information storage. The
limiting case $b_{23} \to b$, $b_{13} \to b$ recovers (i).  

Consider the most general case: $b_{12} \ne b_{23} \ne b_{13}$.  It can
be seen that there are no stationary states using 
\begin{align}
\dot{\rho} = \mathcal{M} \rho,
\end{align}
where $\mathcal{M}$ is the restriction of $\mathbf{H}_S$ and
$\text{L}_{\text{D}}$ to the one-excitation subspace. For the steady state
$\dot{\rho} = 0$, and there exists a nontrivial solution to $\mathcal{M} \rho
=0$ iff $\det(\mathcal{M}) = 0$.  This is calculated to be 
\begin{align}
\det(\mathcal{M}) = - \frac{4 \lambda^3}{27} (b_{12} - b_{23})^2 (b_{12} -
b_{13})^2 (b_{23} - b_{13})^2,
\end{align} 
where $\lambda$ is the nonzero eigenvalue of $(a_{\alpha\beta})$.  This shows
that for $\det(\mathcal{M}) = 0$, two elements of $(b_{\alpha\beta})$ must be
equal, which is case (ii) above.  Note that there are no CDF subspaces for the
general case.   
%%%%%%%%%%%%%%%%%%%%%%%%%%%%%%%%%%%%%%%%%%%%%%%%%%%%%%%%%%%%%%%%%%%%%%%%%%%%%
\subsection{Four qubits}
\label{sec:fqe}
%%%%%%%%%%%%%%%%%%%%%%%%%%%%%%%%%%%%%%%%%%%%%%%%%%%%%%%%%%%%%%%%%%%%%%%%%%%%
\begin{figure}[t]
\begin{center}
\includegraphics[width=7.5cm,height=4.5cm]{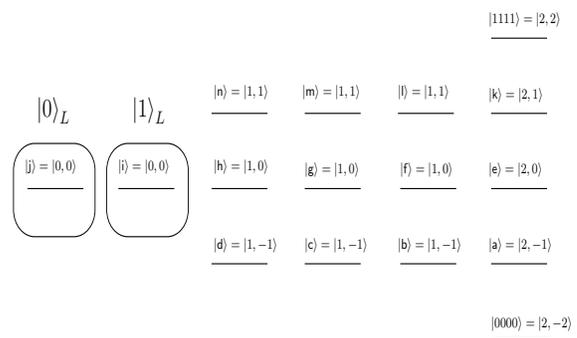}
\end{center}
\caption{\label{fig:fq} Eigenbasis for four qubits, labelled $\ket{J,
  m_J}$, with the logical DFS explicitly labelled.   The splitting of
  the degeneracy due to the non-nearest-neighbor unitary interaction is not  
  included.  }     
\end{figure} 
In order to exploit the collective properties of a system of qubits,
$\tau_2^{-1}$ gives the timescale over which one would expect to be able to
encode information without loss.  However, this timescale might be much faster
than the timescales given by the eigenvalues of $(a_{\alpha\beta})$.  In fact,
for a generic system of dipole-coupled qubits the timescale for unitary
evolution is $\sim 10^8$ times 
faster than the decay timescale~\cite{Brooke07b}.  This difference
in timescales remains for a fully renormalized theory, so for dipole-coupled
qubits DF quantum information will evolve into non-DF states much faster than
the decay rate of the system.

This timescale difference implies that one should encode to protect against
the effect of $\mathbf{H}_S$ before one 
encodes against nonunitary decoherence.  Here, we give explicit examples for
different forms of 
$b_{\alpha \beta}$ and show how one could encode to protect against unitary evolution.  A
significant result presented here is an encoding that protects against an
arbitrary $b_{\alpha \beta}$ in the two-excitation subspace.  The Hamiltonian
\begin{align}
\mathbf{H}_S =  \sum_{\alpha,\beta = 1}^4
b_{\alpha \beta} \hat{\sigma}_{\alpha +} \hat{\sigma}_{\beta -} 
\end{align}  
causes evolution between states that have the same value of $m_J$.  So,
encoding in the strong-collective DFS spanned by $\{ \ket{\sf i}, \ket{\sf j}
\}$ where
\begin{align}
\ket{\sf i} =& \frac{1}{\sqrt{2}}(\ket{01} - \ket{10})(\ket{01} - \ket{10}), \\
\ket{\sf j} =& \frac{1}{\sqrt{12}}(2 \ket{0011} + 2\ket{1100} -\ket{0101} -
\ket{1010} \nonumber \\
&- \ket{0110} - \ket{1001}) 
\end{align}
will not guarantee stable quantum memory.  

Focusing first on the one-excitation subspace, we notice that if $b_{14} =
b_{23}$ and $b_{13} = b_{24}$ then
\begin{align}
\mathbf{H}^{\text{I}}_S = \frac{1}{2}(b_{12} - b_{34}) \ket{\sf a} \bra{\sf d}
+ (b_{24} - b_{23}) \ket{\sf b} \bra{\sf c} + \text{H.c.},
\end{align}
for the basis defined in Fig.~\ref{fig:fq}. So, a logical state encoded across
$\{ \ket{\sf b}, \ket{\sf c} \}$ is not 
acted upon by $\mathbf{H}^{\text{I}}_S$, and satisfies Eq.~\eqref{eq:cdfs}.
Under 
the same conditions, in the two-excitation subspace 
the off-diagonal terms of the Hamiltonian become
\begin{align}
\mathbf{H}^{\text{II}}_S =& \sqrt{\frac{2}{3}}(b_{23} - b_{24})(
  \ket{ \sf e} \bra{\sf i} + \sqrt{2} \ket{ \sf  i} \bra{\sf j} )
  \nonumber \\ 
&+ \frac{\sqrt{2}}{3}(b_{23} - b_{12} + b_{24}
- b_{34}) \ket{ \sf
  e} \bra{\sf j} + \text{H.c.},
\end{align}
which shows that the DFS $\{ \ket{\sf i}, \ket{\sf j} \}$ is coupled to the
symmetric state.  The states that are not acted upon by
$\mathbf{H}^{\text{II}}_S$ are not lowest weight states, and so decay through
the action of $\text{L}_{\text{D}}$.  A further condition is required to
decouple the DF state $\ket{\sf i}$, namely $b_{23} = b_{24}$.  Then, there are
three states that satisfy Eq.~\eqref{eq:cdfs}: $\ket{\sf b}$, $\ket{\sf c}$
and $\ket{\sf i}$, and a logical qubit can be encoded using the most convenient
states for practical applications.  Note that we are interested in storage
times, and do not consider ease of preparation and manipulation.  See
Ref.~\cite{Brooke07a} for one possible method of preparing and manipulating a
logical qubit in a collection of dipole-coupled qubits using
globally-addressed bichromatic incident fields.   

It should be mentioned that since the coefficient of the operator $\ket{\sf j}
\bra{\sf e}$ depends on all values of $b_{\alpha \beta}$, any
perturbation away from $b_{\alpha \beta} \equiv b$ will cause information
encoded in $\ket{\sf j}$ to decohere.  This would prevent the use of an
experimentally controlled $(b_{\alpha \beta})$ being used as a single-qubit
gate on $\{ \ket{\sf i}, \ket{\sf j} \}$.    
  
We now relax the constraints on the system Hamiltonian, and consider 
arbitrary values of $b_{\alpha \beta}$.  We concentrate
on the two-excitation subspace.  The Hamiltonian can be split into two
parts 
\begin{align}
\mathbf{H}^{\text{II}}_S = \mathbf{H}^{\sf fgh}_S + \mathbf{H}^{\sf eij}_S,
\end{align}
where $\mathbf{H}^{\sf fgh}_S$  ($\mathbf{H}^{\sf eij}_S$) acts only on 
$\{ \ket{\sf f}, \ket{\sf g}, \ket{\sf h} \}$ ($\{ \ket{\sf e}, \ket{\sf i},
\ket{\sf j} \}$).  The antisymmetric DF states $\{ \ket{\sf
  i}, \ket{\sf j} \}$ are coupled to the symmetric state $\ket{\sf e}$ that
undergoes nonunitary evolution. The states of interest in the single-particle
basis are 
\begin{align}
\ket{\sf f} &= \frac{1}{\sqrt{2}}(\ket{1100} - \ket{0011}), \\
\ket{\sf g} &= \frac{1}{2}(\ket{0110} + \ket{0101} -\ket{1010} -\ket{1001}), \\
\ket{\sf h} &= \frac{1}{2}(\ket{1001} - \ket{1010} + \ket{0101} -\ket{0110}).
\end{align}
In the collective basis, the zero-logical state in the two-excitation subspace
  that, before a jump occurs,  is immune to an arbitrary $b_{\alpha \beta}$ is 
\begin{align}
\label{eq:zl}
\ket{0}_L = \frac{1}{\sqrt{\Omega^2_1 + \Omega_2^2}}(\Omega_1 
\ket{\sf g} - \Omega_2 \ket{\sf h}),  
\end{align}
for $\Omega_1 = \tfrac{1}{\sqrt{2}}(b_{14} - b_{13} - b_{23} + b_{24})$,  
$\Omega_2 = \tfrac{1}{\sqrt{2}}(b_{13} + b_{14} - b_{23} - b_{24})$, and the
temporal evolution associated with the diagonal terms in $\mathbf{H}^{\sf
  fgh}_S$ has been absorbed into $\ket{\sf g}$ and $\ket{\sf h}$.  
The one-logical state is then a combination of the two remaining eigenstates of
$\mathbf{H}^{\sf fgh}_S$ 
\begin{align}
\label{eq:ol}
\ket{1}_L&= \left \{ \begin{array}{l} 
\frac{1}{\sqrt{2(\Omega_1^2 + \Omega_2^2)}}(\sqrt{\Omega_1^2 + \Omega_2^2}
\ket{\sf f} + \Omega_2 \ket{\sf g} + \Omega_1 \ket{\sf h}),  \\ 
     \frac{1}{\sqrt{2(\Omega_1^2 + \Omega_2^2)}}(
\Omega_2 \ket{\sf g} + \Omega_1 \ket{\sf h} - 
\sqrt{\Omega_1^2 +
       \Omega_2^2}\ket{\sf f}  ). \end{array}
       \right.    
\end{align}
The states $\{\ket{\sf f},\ket{\sf g},\ket{\sf h}\}$ are coupled to each
  other, but are not coupled with states $\{\ket{\sf e},\ket{\sf 
  l},\ket{\sf j}\}$.  Using the encoding given in Eqs.~\eqref{eq:zl}
and~\eqref{eq:ol} means that, up until $\mathbf{J}$ acts on the two-excitation
subspace, the qubit will be immune to an arbitrary
environment induced non-nearest-neighbor evolution.  For the physical example
of dipole-coupled qubits this is surprising, particularly since it was shown
in Ref.~\cite{Car00} that including the dipole-dipole interaction destroyed
any collective-emission behaviour.  Also, for dipole-coupled
qubits in this regime, the timescale for unitary evolution is typically $\sim
10^8$ times quicker than that for decay, so using the encoding in
Eqs.~\eqref{eq:zl} and~\eqref{eq:ol} might have immediate benefits to
applications.  We should emphasize that quantum information will decay within
the $J = 1$ irreps, so the proposed encoding does not support perfect quantum
memory for infinite time.    
%%%%%%%%%%%%%%%%%%%%%%%%%%%%%%%%%%%%%%%%%%%%%%%%%%%%%%%%%%%%%%%%%%%%%%%%%%%%%
\section{Conclusion}
\label{sec:conc}
%%%%%%%%%%%%%%%%%%%%%%%%%%%%%%%%%%%%%%%%%%%%%%%%%%%%%%%%%%%%%%%%%%%%%%%%%%%%
It is well-known that the detrimental effect of bath-induced Hamiltonians 
is not accounted for by requiring $\text{L}_{\text{D}}[\tilde{\rho}] = 
0$, for $\tilde{\rho}$ a DFS~\cite{Lidar98}.  In this paper, we stated a
condition, similar to the conditions in Refs.~\cite{Zan97a,Lidar98,Shab05},
that ensures persistent DF quantum information in the presence
of a non-nearest-neighbor bath-induced system Hamiltonian.  We showed that,
in light of an arbitrary system Hamiltonian, as
the size of the Hilbert space increased, the strong-collective DFS
is the least suitable subspace for quantum information storage.  The most
suitable place to store quantum information in $N$ qubits---if scalability is
important---is probably the subspace with $\sim \tfrac{2}{9} N$ excitations.  We then gave a
timescale over which other methods would have to be 
applied to account for $\mathbf{H}_S$.  We then concentrated on small qubit
systems, giving specific examples for three and four qubits that we hope will
have immediate benefit to applications.  A  particularly interesting result
for four qubits was the encoding that eliminates the need to correct for
$\mathbf{H}_S$ while the qubit remains in the two-excitation subspace. 
%%%%%%%%%%%%%%%%%%%%%%%%%%%%%%%%%%%%%%%%%%%%%%%%%%%%%%%%%%%%%%%%%%%%%%%%%%%% 
\section*{Acknowledgments}
%%%%%%%%%%%%%%%%%%%%%%%%%%%%%%%%%%%%%%%%%%%%%%%%%%%%%%%%%%%%%%%%%%%%%%%%%%%%%
This project was funded by CQCT, Macquarie
University, and iCORE.  PGB acknowledges support and
hospitality during his stay at the University of Calgary where
this work was begun, and David Brooke for helpful discussions. 
\bibliography{dfslimit}    

\begin{thebibliography}{29}
\expandafter\ifx\csname natexlab\endcsname\relax\def\natexlab#1{#1}\fi
\expandafter\ifx\csname bibnamefont\endcsname\relax
  \def\bibnamefont#1{#1}\fi
\expandafter\ifx\csname bibfnamefont\endcsname\relax
  \def\bibfnamefont#1{#1}\fi
\expandafter\ifx\csname citenamefont\endcsname\relax
  \def\citenamefont#1{#1}\fi
\expandafter\ifx\csname url\endcsname\relax
  \def\url#1{\texttt{#1}}\fi
\expandafter\ifx\csname urlprefix\endcsname\relax\def\urlprefix{URL }\fi
\providecommand{\bibinfo}[2]{#2}
\providecommand{\eprint}[2][]{\url{#2}}

\bibitem[{\citenamefont{Palma et~al.}(1996)\citenamefont{Palma, Suominen, and
  Ekert}}]{Pal96}
\bibinfo{author}{\bibfnamefont{G.~M.} \bibnamefont{Palma}},
  \bibinfo{author}{\bibfnamefont{K.-A.} \bibnamefont{Suominen}},
  \bibnamefont{and} \bibinfo{author}{\bibfnamefont{A.~K.} \bibnamefont{Ekert}},
  \bibinfo{journal}{Proc. Roy. Soc. London Ser. A}
  \textbf{\bibinfo{volume}{452}}, \bibinfo{pages}{567} (\bibinfo{year}{1996}).

\bibitem[{\citenamefont{Duan and Guo}(1997)}]{Duan97}
\bibinfo{author}{\bibfnamefont{L.~M.} \bibnamefont{Duan}} \bibnamefont{and}
  \bibinfo{author}{\bibfnamefont{G.~C.} \bibnamefont{Guo}},
  \bibinfo{journal}{Phys. Rev. Lett.} \textbf{\bibinfo{volume}{79}},
  \bibinfo{pages}{1953} (\bibinfo{year}{1997}).

\bibitem[{\citenamefont{Duan and Guo}(1998)}]{Duan98}
\bibinfo{author}{\bibfnamefont{L.~M.} \bibnamefont{Duan}} \bibnamefont{and}
  \bibinfo{author}{\bibfnamefont{G.~C.} \bibnamefont{Guo}},
  \bibinfo{journal}{Phys. Rev. A} \textbf{\bibinfo{volume}{57}},
  \bibinfo{pages}{737} (\bibinfo{year}{1998}).

\bibitem[{\citenamefont{Zanardi and Rasetti}(1997{\natexlab{a}})}]{Zan97a}
\bibinfo{author}{\bibfnamefont{P.}~\bibnamefont{Zanardi}} \bibnamefont{and}
  \bibinfo{author}{\bibfnamefont{M.}~\bibnamefont{Rasetti}},
  \bibinfo{journal}{Phys. Rev. Lett.} \textbf{\bibinfo{volume}{79}},
  \bibinfo{pages}{3306} (\bibinfo{year}{1997}{\natexlab{a}}).

\bibitem[{\citenamefont{Zanardi and Rasetti}(1997{\natexlab{b}})}]{Zan97b}
\bibinfo{author}{\bibfnamefont{P.}~\bibnamefont{Zanardi}} \bibnamefont{and}
  \bibinfo{author}{\bibfnamefont{M.}~\bibnamefont{Rasetti}},
  \bibinfo{journal}{Mod. Phys. Lett B} \textbf{\bibinfo{volume}{11}},
  \bibinfo{pages}{1085} (\bibinfo{year}{1997}{\natexlab{b}}).

\bibitem[{\citenamefont{Zanardi}(1998)}]{Zan98}
\bibinfo{author}{\bibfnamefont{P.}~\bibnamefont{Zanardi}},
  \bibinfo{journal}{Phys. Rev. A} \textbf{\bibinfo{volume}{57}},
  \bibinfo{pages}{3276} (\bibinfo{year}{1998}).

\bibitem[{\citenamefont{Lidar et~al.}(1998)\citenamefont{Lidar, Chuang, and
  Whaley}}]{Lidar98}
\bibinfo{author}{\bibfnamefont{D.~A.} \bibnamefont{Lidar}},
  \bibinfo{author}{\bibfnamefont{I.~L.} \bibnamefont{Chuang}},
  \bibnamefont{and} \bibinfo{author}{\bibfnamefont{K.~B.}
  \bibnamefont{Whaley}}, \bibinfo{journal}{Phys. Rev. Lett.}
  \textbf{\bibinfo{volume}{81}}, \bibinfo{pages}{2594} (\bibinfo{year}{1998}).

\bibitem[{\citenamefont{Knill et~al.}(2000)\citenamefont{Knill, Laflamme, and
  Viola}}]{Knill00}
\bibinfo{author}{\bibfnamefont{E.}~\bibnamefont{Knill}},
  \bibinfo{author}{\bibfnamefont{R.}~\bibnamefont{Laflamme}}, \bibnamefont{and}
  \bibinfo{author}{\bibfnamefont{L.}~\bibnamefont{Viola}},
  \bibinfo{journal}{Phys. Rev. Lett.} \textbf{\bibinfo{volume}{84}},
  \bibinfo{pages}{2525} (\bibinfo{year}{2000}).

\bibitem[{\citenamefont{Lidar and Whaley}(2003)}]{Lidar03}
\bibinfo{author}{\bibfnamefont{D.~A.} \bibnamefont{Lidar}} \bibnamefont{and}
  \bibinfo{author}{\bibfnamefont{K.~B.} \bibnamefont{Whaley}},
  \emph{\bibinfo{title}{Decoherence-Free Subspaces and Subsystems}}
  (\bibinfo{publisher}{Springer-Verlag, Berlin}, \bibinfo{year}{2003}), vol.
  \bibinfo{volume}{622}, pp. \bibinfo{pages}{83--120}.

\bibitem[{\citenamefont{Kwiat et~al.}(2000)\citenamefont{Kwiat, Berglund,
  Altepeter, and White}}]{Kw00}
\bibinfo{author}{\bibfnamefont{P.~G.} \bibnamefont{Kwiat}},
  \bibinfo{author}{\bibfnamefont{A.~J.} \bibnamefont{Berglund}},
  \bibinfo{author}{\bibfnamefont{J.~B.} \bibnamefont{Altepeter}},
  \bibnamefont{and} \bibinfo{author}{\bibfnamefont{A.~G.} \bibnamefont{White}},
  \bibinfo{journal}{Science} \textbf{\bibinfo{volume}{290}},
  \bibinfo{pages}{498} (\bibinfo{year}{2000}).

\bibitem[{\citenamefont{Kielpinski et~al.}(2001)\citenamefont{Kielpinski,
  Meyer, Rowe, Sackett, Itano, Munroe, and Wineland}}]{Kiel01}
\bibinfo{author}{\bibfnamefont{D.}~\bibnamefont{Kielpinski}},
  \bibinfo{author}{\bibfnamefont{V.}~\bibnamefont{Meyer}},
  \bibinfo{author}{\bibfnamefont{M.~A.} \bibnamefont{Rowe}},
  \bibinfo{author}{\bibfnamefont{C.~A.} \bibnamefont{Sackett}},
  \bibinfo{author}{\bibfnamefont{W.~M.} \bibnamefont{Itano}},
  \bibinfo{author}{\bibfnamefont{C.}~\bibnamefont{Munroe}}, \bibnamefont{and}
  \bibinfo{author}{\bibfnamefont{D.~J.} \bibnamefont{Wineland}},
  \bibinfo{journal}{Science} \textbf{\bibinfo{volume}{291}},
  \bibinfo{pages}{1013} (\bibinfo{year}{2001}).

\bibitem[{\citenamefont{Viola et~al.}(2001)\citenamefont{Viola, Fortunato,
  Pravia, Knill, Laflamme, and Cory}}]{Lo01}
\bibinfo{author}{\bibfnamefont{L.}~\bibnamefont{Viola}},
  \bibinfo{author}{\bibfnamefont{E.~M.} \bibnamefont{Fortunato}},
  \bibinfo{author}{\bibfnamefont{M.~A.} \bibnamefont{Pravia}},
  \bibinfo{author}{\bibfnamefont{E.}~\bibnamefont{Knill}},
  \bibinfo{author}{\bibfnamefont{R.}~\bibnamefont{Laflamme}}, \bibnamefont{and}
  \bibinfo{author}{\bibfnamefont{D.~G.} \bibnamefont{Cory}},
  \bibinfo{journal}{Science} \textbf{\bibinfo{volume}{293}},
  \bibinfo{pages}{2059} (\bibinfo{year}{2001}).

\bibitem[{\citenamefont{Bourennane et~al.}(2004)\citenamefont{Bourennane, Eibl,
  Gaertner, Kurtsiefer, Cabello, and Weinfurter}}]{Bou04}
\bibinfo{author}{\bibfnamefont{M.}~\bibnamefont{Bourennane}},
  \bibinfo{author}{\bibfnamefont{M.}~\bibnamefont{Eibl}},
  \bibinfo{author}{\bibfnamefont{S.}~\bibnamefont{Gaertner}},
  \bibinfo{author}{\bibfnamefont{C.}~\bibnamefont{Kurtsiefer}},
  \bibinfo{author}{\bibfnamefont{A.}~\bibnamefont{Cabello}}, \bibnamefont{and}
  \bibinfo{author}{\bibfnamefont{H.}~\bibnamefont{Weinfurter}},
  \bibinfo{journal}{Phys. Rev. Lett} \textbf{\bibinfo{volume}{92}},
  \bibinfo{pages}{107901} (\bibinfo{year}{2004}).

\bibitem[{\citenamefont{Brooke}(2007)}]{Brooke07a}
\bibinfo{author}{\bibfnamefont{P.~G.} \bibnamefont{Brooke}},
  \bibinfo{journal}{Phys. Rev. A} \textbf{\bibinfo{volume}{75}},
  \bibinfo{pages}{022320} (\bibinfo{year}{2007}).

\bibitem[{\citenamefont{Shabani and Lidar}(2005)}]{Shab05}
\bibinfo{author}{\bibfnamefont{A.}~\bibnamefont{Shabani}} \bibnamefont{and}
  \bibinfo{author}{\bibfnamefont{D.~A.} \bibnamefont{Lidar}},
  \bibinfo{journal}{Phys. Rev. A} \textbf{\bibinfo{volume}{72}},
  \bibinfo{pages}{042303} (\bibinfo{year}{2005}).

\bibitem[{\citenamefont{Brooke et~al.}(2008)\citenamefont{Brooke, Marzlin,
  Cresser, and Sanders}}]{Brooke07b}
\bibinfo{author}{\bibfnamefont{P.~G.} \bibnamefont{Brooke}},
  \bibinfo{author}{\bibfnamefont{K.-P.} \bibnamefont{Marzlin}},
  \bibinfo{author}{\bibfnamefont{J.~D.} \bibnamefont{Cresser}},
  \bibnamefont{and} \bibinfo{author}{\bibfnamefont{B.~C.}
  \bibnamefont{Sanders}}, \bibinfo{journal}{Phys. Rev. A}
  \textbf{\bibinfo{volume}{77}}, \bibinfo{pages}{033844}
  (\bibinfo{year}{2008}).

\bibitem[{\citenamefont{Davies}(1994)}]{Da94}
\bibinfo{author}{\bibfnamefont{G.}~\bibnamefont{Davies}},
  \emph{\bibinfo{title}{Properties and Growth of Diamond}}
  (\bibinfo{publisher}{IEE/INSPEC, London, Vol. 9}, \bibinfo{year}{1994}).

\bibitem[{\citenamefont{Lindblad}(1976)}]{Lind}
\bibinfo{author}{\bibfnamefont{G.}~\bibnamefont{Lindblad}},
  \bibinfo{journal}{Commun. Math. Phys.} \textbf{\bibinfo{volume}{48}},
  \bibinfo{pages}{119} (\bibinfo{year}{1976}).

\bibitem[{\citenamefont{Khodjasteh and Lidar}(2002)}]{Lidar02a}
\bibinfo{author}{\bibfnamefont{K.}~\bibnamefont{Khodjasteh}} \bibnamefont{and}
  \bibinfo{author}{\bibfnamefont{D.~A.} \bibnamefont{Lidar}},
  \bibinfo{journal}{Phys. Rev. Lett.} \textbf{\bibinfo{volume}{89}},
  \bibinfo{pages}{197904} (\bibinfo{year}{2002}).

\bibitem[{\citenamefont{Bacon et~al.}(1999)\citenamefont{Bacon, Lidar, and
  Whaley}}]{Bacon99}
\bibinfo{author}{\bibfnamefont{D.}~\bibnamefont{Bacon}},
  \bibinfo{author}{\bibfnamefont{D.~A.} \bibnamefont{Lidar}}, \bibnamefont{and}
  \bibinfo{author}{\bibfnamefont{K.~B.} \bibnamefont{Whaley}},
  \bibinfo{journal}{Phys. Rev. A} \textbf{\bibinfo{volume}{60}},
  \bibinfo{pages}{1944} (\bibinfo{year}{1999}).

\bibitem[{\citenamefont{Viola et~al.}(1999)\citenamefont{Viola, Lloyd, and
  Knill}}]{Viola99}
\bibinfo{author}{\bibfnamefont{L.}~\bibnamefont{Viola}},
  \bibinfo{author}{\bibfnamefont{S.}~\bibnamefont{Lloyd}}, \bibnamefont{and}
  \bibinfo{author}{\bibfnamefont{E.}~\bibnamefont{Knill}},
  \bibinfo{journal}{Phys. Rev. Lett.} \textbf{\bibinfo{volume}{83}},
  \bibinfo{pages}{4888} (\bibinfo{year}{1999}).

\bibitem[{\citenamefont{Wu et~al.}(2002)\citenamefont{Wu, Byrd, and
  Lidar}}]{Wu002}
\bibinfo{author}{\bibfnamefont{L.-A.} \bibnamefont{Wu}},
  \bibinfo{author}{\bibfnamefont{M.~S.} \bibnamefont{Byrd}}, \bibnamefont{and}
  \bibinfo{author}{\bibfnamefont{D.~A.} \bibnamefont{Lidar}},
  \bibinfo{journal}{Phys. Rev. Lett.} \textbf{\bibinfo{volume}{89}},
  \bibinfo{pages}{127901} (\bibinfo{year}{2002}).

\bibitem[{\citenamefont{Byrd and Lidar}(2003)}]{Byrd03}
\bibinfo{author}{\bibfnamefont{M.~S.} \bibnamefont{Byrd}} \bibnamefont{and}
  \bibinfo{author}{\bibfnamefont{D.~A.} \bibnamefont{Lidar}},
  \bibinfo{journal}{Phys. Rev. A} \textbf{\bibinfo{volume}{67}},
  \bibinfo{pages}{012324} (\bibinfo{year}{2003}).

\bibitem[{\citenamefont{Byrd et~al.}(2005)\citenamefont{Byrd, Lidar, Wu, and
  Zanardi}}]{Byrd05}
\bibinfo{author}{\bibfnamefont{M.~S.} \bibnamefont{Byrd}},
  \bibinfo{author}{\bibfnamefont{D.~A.} \bibnamefont{Lidar}},
  \bibinfo{author}{\bibfnamefont{L.-A.} \bibnamefont{Wu}}, \bibnamefont{and}
  \bibinfo{author}{\bibfnamefont{P.}~\bibnamefont{Zanardi}},
  \bibinfo{journal}{Phys. Rev. A} \textbf{\bibinfo{volume}{71}},
  \bibinfo{pages}{052301} (\bibinfo{year}{2005}).

\bibitem[{\citenamefont{Gilmore}(2005)}]{Gilmore}
\bibinfo{author}{\bibfnamefont{R.}~\bibnamefont{Gilmore}},
  \emph{\bibinfo{title}{Lie Groups, Lie Algebras, and Some of Their
  Applications}} (\bibinfo{publisher}{Dover}, \bibinfo{address}{New York},
  \bibinfo{year}{2005}).

\bibitem[{\citenamefont{Kempe et~al.}(2001)\citenamefont{Kempe, Bacon, Lidar,
  and Whaley}}]{Kempe01}
\bibinfo{author}{\bibfnamefont{J.}~\bibnamefont{Kempe}},
  \bibinfo{author}{\bibfnamefont{D.}~\bibnamefont{Bacon}},
  \bibinfo{author}{\bibfnamefont{D.~A.} \bibnamefont{Lidar}}, \bibnamefont{and}
  \bibinfo{author}{\bibfnamefont{K.~B.} \bibnamefont{Whaley}},
  \bibinfo{journal}{Phys. Rev. A} \textbf{\bibinfo{volume}{63}},
  \bibinfo{pages}{042307} (\bibinfo{year}{2001}).

\bibitem[{\citenamefont{Neilsen and Chuang}(2000)}]{Neil00}
\bibinfo{author}{\bibfnamefont{M.~A.} \bibnamefont{Neilsen}} \bibnamefont{and}
  \bibinfo{author}{\bibfnamefont{I.~L.} \bibnamefont{Chuang}},
  \emph{\bibinfo{title}{Quantum computation and Quantum Information}}
  (\bibinfo{publisher}{Cambridge University Press},
  \bibinfo{address}{Cambridge, England}, \bibinfo{year}{2000}).

\bibitem[{\citenamefont{Meijer et~al.}(2006)\citenamefont{Meijer, Vogel,
  Burchard, Rangelow, Bischoff, Wrachtrup, Domhan, Jelezko, Schnitzler, Schulz
  et~al.}}]{Me05a}
\bibinfo{author}{\bibfnamefont{J.}~\bibnamefont{Meijer}},
  \bibinfo{author}{\bibfnamefont{T.}~\bibnamefont{Vogel}},
  \bibinfo{author}{\bibfnamefont{B.}~\bibnamefont{Burchard}},
  \bibinfo{author}{\bibfnamefont{I.}~\bibnamefont{Rangelow}},
  \bibinfo{author}{\bibfnamefont{L.}~\bibnamefont{Bischoff}},
  \bibinfo{author}{\bibfnamefont{J.}~\bibnamefont{Wrachtrup}},
  \bibinfo{author}{\bibfnamefont{M.}~\bibnamefont{Domhan}},
  \bibinfo{author}{\bibfnamefont{F.}~\bibnamefont{Jelezko}},
  \bibinfo{author}{\bibfnamefont{W.}~\bibnamefont{Schnitzler}},
  \bibinfo{author}{\bibfnamefont{S.~A.} \bibnamefont{Schulz}},
  \bibnamefont{et~al.}, \bibinfo{journal}{Appl.Phys.A}
  \textbf{\bibinfo{volume}{83}}, \bibinfo{pages}{321} (\bibinfo{year}{2006}).

\bibitem[{\citenamefont{Carmichael and Kim}(2000)}]{Car00}
\bibinfo{author}{\bibfnamefont{H.}~\bibnamefont{Carmichael}} \bibnamefont{and}
  \bibinfo{author}{\bibfnamefont{K.}~\bibnamefont{Kim}},
  \bibinfo{journal}{Opt.Commun.} \textbf{\bibinfo{volume}{179}},
  \bibinfo{pages}{417} (\bibinfo{year}{2000}).

\end{thebibliography}
\end{document}